\documentclass[conferene]{IEEEtran}
\usepackage[utf8]{inputenc}
\usepackage{amsmath,amsfonts,amsthm}
\usepackage{algorithm,algorithmic}
\usepackage{booktabs}
\usepackage{bm}
\usepackage{comment}
\usepackage{chngpage}
\usepackage{enumerate}
\usepackage{paralist}
\usepackage{extramarks}
\usepackage{fancyhdr}
\usepackage{graphicx,float,wrapfig}
\usepackage{lastpage}
\usepackage{setspace}
\usepackage{soul,color}
\usepackage{subcaption}
\usepackage{xspace}
\usepackage{indentfirst,url}

\usepackage{nohyperref}

\makeatletter
\DeclareRobustCommand\onedot{\futurelet\@let@token\@onedot}
\def\@onedot{\ifx\@let@token.\else.\null\fi\xspace}

\def\st{$s-t~$}

\def\PP{\mathcal P}
\makeatother

\theoremstyle{definition}
\newtheorem{defn}{Definition}

\theoremstyle{plain}
\newtheorem{thm}{Theorem}
\newtheorem{lem}{Lemma}

\theoremstyle{remark}
\newtheorem*{rmk}{Remark}

\newcommand*{\rom}[1]{\expandafter\@slowromancap\romannumeral #1@}

\title{On the Robustness of Distributed Computing Networks}
\author{Jianan Zhang, Hyang-Won Lee, and Eytan Modiano
\thanks{
Part of the material in this paper was presented at the International Conference on the Design of reliable communication networks (DRCN), 2019.

J. Zhang and E. Modiano are with the Laboratory for Information and Decision Systems, Massachusetts Institute of Technology, USA.

H.-W. Lee is with the Department of Software, Konkuk University, Republic of Korea.

This work was supported by DTRA grants HDTRA1-13-1-0021 and HDTRA1-14-1-0058, and NSF grant CNS-1617091. The work of Hyang-Won Lee was supported by the National Research Foundation of Korea (NRF) grant funded by the Korea government (MSIT) (No.2018R1D1A1B07048388).
}}

\begin{document}
\maketitle
\begin{abstract}
Traffic flows in a distributed computing network require both transmission and processing, and can be interdicted by removing either communication or computation resources. We study the robustness of a distributed computing network under the failures of communication links and computation nodes. We define cut metrics that measure the connectivity, and show a non-zero gap between the maximum flow and the minimum cut. Moreover, we study a network flow interdiction problem that minimizes the maximum flow by removing communication and computation resources within a given budget. We develop mathematical programs to compute the optimal interdiction, and polynomial-time approximation algorithms that achieve near-optimal interdiction in simulation.
\end{abstract}

\section{Introduction}
Cloud computing has been growing rapidly in recent years. For example, over one millon servers have been deployed for Amazon Web Service, which generates billions of dollars in revenue each year and grew by over 40 percent in revenue in 2018. Cloud networks, and computing networks in general, facilitate agile, reliable and cost effective implementations for a variety of applications. The robustness of computing networks is essential for web access, online database, video streaming, among other applications deployed in the cloud.

Network flows in a computing network rely on both communication resources for transmission and computation resources for processing. The unavailability of either type of resources may lead to the failure of flows. Hundreds of thousands of websites were down due to the computation resource failure in a data center for Amazon Web Service \cite{aws}. In 2006, Internet services in Asia were disrupted by communication failure due to the broken of submarine cables by earthquake \cite{earth}. 


The dependence of network flow on various types of resources brings challenges to the reliability of a computing network \cite{azodolmolky2013cloud,gill2011understanding, charikar2018multi}. Previous research proposed new computing network architectures to improve reliability \cite{mohammed2017failover, malik2011adaptive, cheraghlou2016survey}, and developed models to study failure cascading and protection strategies \cite{rao2012cloud,chauvel2015evaluating}. However, limited works focus on the rigorous analysis of network flow reduction under the failures of network resources, which is a key metric for computing network performance and is the focus of this paper.

Flow interdiction problems have been extensively studied based on the classical flow network model. The problem of minimizing the maximum flow by removing network links within a budget is strongly NP-hard \cite{wood1993deterministic}. Integer linear programs were developed to compute the optimal interdiction \cite{wood1993deterministic}. Approximation hardness results and a $2(n-1)$-approximation algorithm was developed in \cite{chestnut2017hardness}. A pseudoapproximation algorithm was developed in \cite{burch2003decomposition} based on linear programming relaxation, and developed in \cite{mccormick2014discrete} with faster combinatorial algorithm implementation. NP-hardness result and a polynomial-time approximation scheme were developed for network flow interdiction on planer graphs \cite{phillips1993network,zenklusen2010network}.

In a traditional flow network, the maximum flow between a source-destination ($s-t$) pair equals the minimum cut, which is the minimum-capacity link removals that disconnect the \st pair \cite{ford1956maximal}. In a computing network, we show that there is a non-zero gap between the maximum flow and the minimum cut. 
The non-zero gap between maximum flow and minimum cut exists in a wide range of network interdiction scenarios, abstracted by the shared risk group model \cite{coudert2007shared}, where a single failure event may destroy multiple network components. For example, in layered communication networks, such as IP-over-WDM networks, the failure of a physical link may affect multiple logical links, and the maximum number of failure-disjoint paths could be smaller than the minimum number of physical link failures that induce a cut \cite{lee2011cross, hu2003diverse}. In geographically correlated failure models \cite{neumayer2011assessing, agarwal2013resilience, msongaleli2016disaster}, one geographical failure affects multiple nodes and links, and the minimum cut can also be greater than the maximum flow \cite{neumayer2015geographic}. Although seemingly unrelated, we show that a computing network can be analyzed using a layered graph where link failures are coupled, and thus the gap exists.

The main contributions of this paper are as follows. We propose a model for a computing network to characterize the dependence of network flow on both communication and computation resources. The model facilitates the analysis of computing network robustness, by integrating the modeling of the computation resource to a classical graph model. By extending the classical cut metric for a graph, we define cut metrics that characterize computing network robustness under the failures of communication and computation resources. We prove the computation complexity, and develop integer programs and approximation algorithms to compute the minimum cuts. Moreover, we formulate a maximum flow interdiction problem, where the objective is to minimize the maximum \st flow by removing network resources within a given budget. We prove the computation complexity, and develop exact and approximation algorithms to compute optimal interdiction strategies. A preliminary version of this paper appeared in \cite{zhang2019robustness}.

The rest of this paper is organized as follows. In Section \ref{sc6:model}, we introduce the model for a distributed computing network, and define cut metrics to evaluate the network robustness. In Section \ref{sc6:eval}, we develop algorithms to evaluate the maximum flow and minimum cuts. In Section \ref{sc6:inter}, we formulate and solve a maximum flow interdiction problem with an interdiction budget. Section \ref{sc6:num} provides numerical results. Section \ref{sc6:conclusion} concludes the paper. 

\section{Model}
\label{sc6:model}
In this section, we develop a model for a distributed computing network, and define metrics for network robustness.

A distributed computing network is modeled by a directed graph $G(V,E)$, where $V$ denotes the set of forwarding and computation nodes, and $E$ denotes the set of communication links. Computation nodes can process and forward packets, while forwarding nodes can only forward packets. A computation node $u \in V$ has processing capacity $\mu_u$, and a communication link $(u,v) \in E$ has transmission capacity $\mu_{uv}$. 

Unlike the traditional data network where flows require minimal fixed computation tasks such as routing table lookup and checksum, flows in the distributed computing network can require vastly different computation resources, and hence computation capacities at servers (as well as communication bandwidth) are essential to process traffic. The classical robustness metric such as minimum cut is not able to capture the robustness of such a computing network. We extend classical flow and cut metrics to computing networks, to characterize the need to incorporate both communication and computation resources in network operation.

We first define \emph{computation path} which supports both the processing and the delivery of data packets in the network.
\begin{defn}
  A \emph{computation path} $(P,w)$ from a source $s$ to a destination $t$ is characterized by a sequence of connected edges (and their end nodes) $P$ that start at $s$ and end at $t$, and includes a computation node $w \in P$.
\end{defn}

A network flow consists of packets that are originated at a source, processed at one or more computation nodes, and delivered to a destination. A flow can be decomposed into computation paths. We illustrate an \st flow and computation paths decomposition using an example in Fig. \ref{fig6:eg}, where computation nodes are illustrated by squares and forwarding nodes are illustrated by circles, and the numbers represent capacities. The maximum \st flow is four, and can be decomposed into one unit flow on each of the four computation paths $(\{s-u_1-v_1-t\},u_1),(\{s-u_1-v_1-t\},v_1), (\{s-u_2-v_2-t\},u_2),(\{s-u_2-v_2-t\},v_2)$.

\begin{figure}[h]
\centering
\includegraphics[width=.59\linewidth]{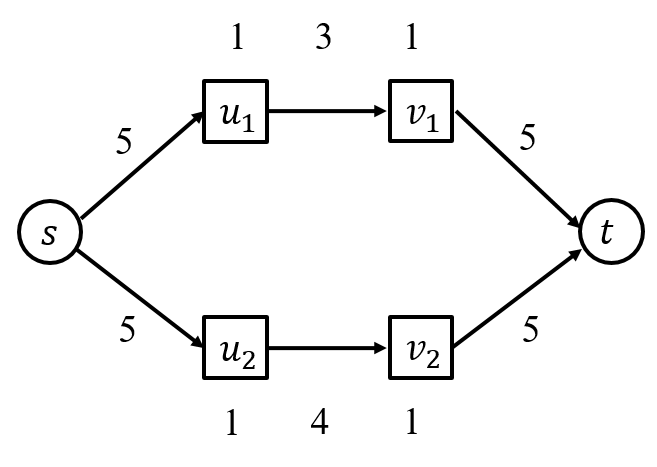}
\caption{An example of a distributed computing network.}
\label{fig6:eg}
\end{figure}

In order to reduce the flow carried by a computation path to zero, either any communication link or the computation resource in the path should be removed. Note that we consider the removal of computation resources without removing the node, i.e., the node can still forward packets without processing them.

In general, there are multiple computation paths from a source to a destination. To interdict the flow, a combination of communication and computation resources can be removed. We next define cuts that measure the connectivity of a pair of nodes in a computing network.

\begin{defn}
  A \emph{communication cut} for an \st pair is a set of communication links $E_c$ such that the \st flow is reduced to zero after removing $E_c$.
\end{defn}

Note that a communication cut can be different from the classical edge cut, since a path with zero computation resource does not need to be disconnected.

\begin{defn}
  A \emph{computation cut} for an \st pair is a set of computation nodes $V_c$ such that the \st flow is reduced to zero after removing the computation resources at $V_c$.
\end{defn}

\begin{defn}
  A \emph{joint communication and computation cut} (abbr. \emph{joint cut}) for an \st pair is a set of communication links $E_c$ and computation nodes $V_c$ such that the \st flow is reduced to zero after removing $E_c$ and computation resources at $V_c$.
\end{defn}

We illustrate these cuts using the example in Fig. \ref{fig6:eg}. Edges $\{(u_1,v_1), (u_2,v_2)\}$ form a communication cut, since $s$ and $t$ are disconnected after removing the two links. Nodes $\{u_1, v_1, u_2, v_2\}$ form a computation cut, since no flow can be processed after removing the computation resources at the four computation nodes. The union of edge $\{(u_1,v_1)\}$ and nodes $\{u_2,v_2\}$ is a joint cut, since the upper path is disconnected after removing $(u_1,v_1)$, and the lower path cannot process flow after removing the computation resources at nodes $\{u_2,v_2\}$. 

To simplify the analysis for network robustness, we assume that \emph{all flows have the same resource requirement,} i.e., every unit flow requires a fixed amount of communication and computation resources.
The identical resource requirement of flows can be justified by the statistical multiplexing of individual flows in a network, although individual flows for different applications may have different resource requirements. For example, video streaming is communication intensive, while search is computation intensive. By normalizing units and ignoring flow scaling, we further assume that \emph{every unit flow requires one unit of computation resource for processing, and outputs one unit processed flow.} Under this assumption, one unit flow on a computation path occupies one unit communication resource at every link along the path, and one unit computation resource at the computation node.

Before developing algorithms to evaluate the maximum flow and the minimum cuts, we prove the complexity of computing the cut metrics.
First, we show the hardness of evaluating the minimum communication cut, whose proof is in the Appendix.
\begin{lem} \label{th:comm}
  Computing the minimum communication cut for an \st pair is NP-hard, if there is more than one computation node.
\end{lem}

Lemma \ref{th:comm} implies that computing the minimum joint cut is NP-hard, since the minimum communication cut can be viewed as a special case of the minimum joint cut when the computation resources are abundant at computation nodes.
\begin{thm} \label{th:joint}
  Computing the minimum joint communication and computation cut for an \st pair is NP-hard, if there is more than one computation node.
\end{thm}

\section{Computation of max-flow and min-cuts}
\label{sc6:eval}
In this section, we study the computation of the maximum flow and minimum cuts for a source-destination pair. We develop polynomial-time algorithms to evaluate the maximum flow and the minimum computation cut, and integer programs to evaluate the minimum communication cut and the minimum joint cut. In Section \ref{sc6:path}, we develop mathematical programs to evaluate the maximum flow and the minimum cut using path-based formulation, which are intuitive but have an exponential number of variables or constraints. In Section \ref{sc6:layer}, we develop a layered graph representation to simplify their computations, and develop mathematical programs of polynomial sizes. Finally, in Section \ref{sc6:mfmc}, we study the gap between the maximum flow and the minimum cuts.

\subsection{Path-based formulations} \label{sc6:path}
We first develop mathematical programs to compute the maximum flow and minimum cuts using path-based formulations. While the formulations have an exponential number of variables or constraints, they highlight the connections between flow and cuts in a computing network to those in a classical flow network.

We formulate a linear program to compute the maximum flow in a computing network. Let $\PP$ denote the set of \st paths. 
Let $x_{P,w}$ denote the amount of flow transmitted through path $P$ and processed at a computation node $w \in P$. The maximum flow can be computed by the following linear program.
\begin{align}
  \max ~~~& ~~~\sum_{P \in \PP, w \in P} x_{P,w} \label{eq:maxflow} \\
  \text{s.t.} ~~~ & \sum_{P \in \PP, w \in P: (u,v) \in P} x_{P,w} \leq \mu_{uv}, ~~~\forall (u,v) \in E, \label{eq:comm_cap} \\
  & \sum_{P \in \PP: w \in P} x_{P,w} \leq \mu_w, ~~~\forall w \in V, \label{eq:compcap} \\
  & x_{P,w} \geq 0, ~~~ \forall {P \in \PP, w \in P}. \nonumber
\end{align}

The communication capacity constraints are guaranteed by \eqref{eq:comm_cap}, and the computation capacity constraints are guaranteed by \eqref{eq:compcap}, by restricting the total amount of flow transmitted by a link or processed at a computation node. The objective is to maximize the total amount of flow supported by the computation paths. 

We then develop an integer program to evaluate the minimum joint communication and computation cut using the path-based formulation. Indicator variable $y_{uv}$ represents whether link $(u,v)$ is removed. Indicator variable $y_w$ represents whether the computation resource at node $w$ is removed. Constraint \eqref{eq:cut2} guarantees that for each path, either one of the link is removed, or all the computation resources are removed.

\begin{align}
  \min ~~~& ~~~\sum_{(u,v) \in E} \mu_{uv} y_{uv} + \sum_{w \in V} \mu_w y_w \label{eq:mincut}\\
  \text{s.t.} ~~~ & \sum_{(u,v) \in P} y_{uv} + y_w \geq 1, ~~~\forall P \in \PP, w \in P \label{eq:cut2} \\
  & y_{uv} \in \{0,1\}, ~~~ \forall (u,v) \in E \nonumber \\
  & y_w \in \{0,1\}, ~~~ \forall w \in V. \nonumber
\end{align}

Pure communication or computation cuts can be obtained by the above integer program with additional constraints. A minimum communication cut can be obtained by setting $y_w = 0, \forall w \in V$. A minimum computation cut can be obtained by setting $y_{uv} = 0, \forall (u,v) \in E$.

The number of paths $|\PP|$ can be exponential in the size of the network. Both the linear program \eqref{eq:maxflow} and the integer program \eqref{eq:mincut} have exponential sizes. Compared with the classical maximum flow and minimum cut formulations, the main difference is that a computation path in the computing network depends on a computation node in addition to a sequence of connected links. The coupling of constraints by the computation nodes brings challenges to the evaluation of the flow and cut metrics.

\subsection{Layered graph formulations}\label{sc6:layer}
To address the challenges, we develop a layered graph representation to simplify the evaluation of flow and cuts. Based on the layered graph, in Sections \ref{sc6:flow} and \ref{sc6:cut}, we develop modified mathematical programs with a polynomial number of variables and constraints to evaluate the maximum flow and the minimum cuts, respectively.

We consider a two-layer graph, where every layer has the same topology as the original graph. An edge connects the two copies of each computation node across the two layers. Unprocessed flows are transmitted thought links in the upper layer $G(V,E)$, while processed flows are transmitted in the lower layer $G'(V',E')$. Flows across the two layers represent processing at computation nodes. For example, in Fig. \ref{fig6:layer}, a flow is transmitted through $(s,u)$, processed at $u$, and then transmitted through $(u,v)$ and $(v,t)$. In the layered graph, unprocessed flow is transmitted through $(s,u)$ in the upper layer, then transmitted through $(u,u')$, which represents the processing at $u$, and finally transmitted through $(u',v')$ and $(v',t')$ in the lower layer. Every flow from $s$ to $t$ and processed at computation nodes in the original graph can be represented by a flow from $s$ to $t'$ in the layered graph. We next show that the network resource failures that disconnect $(s,t)$ in the computing network can be mapped to failures that disconnect $(s,t')$ in the layered graph.

\begin{lem}\label{th:layercut}
  Let $S$ be an \st cut in the computing network. In the layered graph, removing edges $S' = \{(u,v), (u',v')| (u,v) \in S\} \cup \{(w,w') | w \in S\}$ disconnects $s$ and $t'$.
\end{lem}

\begin{proof}
  We prove by contradiction. Suppose that a path $P$ exists between $s$ and $t'$ in the layered graph after removing $S'$. The path $P$ must contain a link from the upper layer to the lower layer, denoted by $(a, a')$. There is a path $P_1$ from $s$ to $a$ in the upper layer, and a path $P'_2$ from $a'$ to $t'$ in the lower layer. Let $P_2 =  \{(u,v)|(u',v') \in P'_2\}$. Since none of the edges $P_1 \cup P'_2$ belong to cut $S'$, none of the edges $P_1 \cup P_2$ belong to cut $S$, under the construction of $S'$.

  In the computing network, there is a path $P_1$ from $s$ to $a$, and a path $P_2$ from $a$ to $t$. Moreover, the computation resource at $a$ is not removed, since $(a, a')$ remains in the layered graph. The path $(P_1 \cup P_2 , a)$ is a computation path from $s$ to $t$, which contradicts with the fact that $S$ is an \st cut.
\end{proof}

\begin{figure}[h]
\centering
\includegraphics[width=\linewidth]{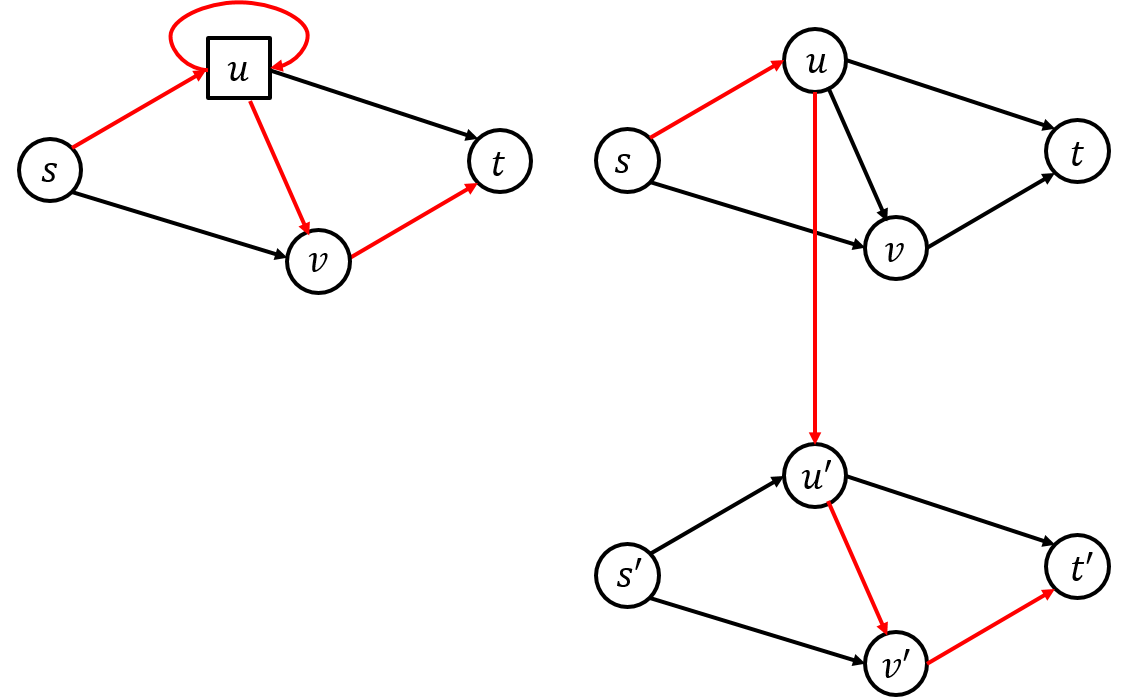}
\caption{Flows in the original and layered graphs.}
\label{fig6:layer}
\end{figure}

\subsubsection{Evaluation of maximum flow}
\label{sc6:flow}
Flow conservation holds in the layered graph, since communication and computation units are normalized and flow scalings are ignored. The difference between a flow in the layered graph and the classical network flow is that the sum of flows on the two copies of a link should not exceed the transmission capacity. Let $\tilde E$ denote the union of the set of links in the layered graph and a link from $t'$ to $s$ that has an infinite capacity. Let $\tilde V = V \cup V'$ denote the set of nodes in the layered graph. Let $f_{e}$ denote the amount of flow on link $e$ in the layered graph. The maximum flow from $s$ to $t'$ can be computed using the following linear program.

\begin{align}
  \max ~& ~~~f_{t's} \label{eq:flowpoly}\\
  \text{s.t.} ~ & \sum_{u \in \tilde V:(u,v) \in \tilde E}\hspace{-3mm} f_{uv} - \sum_{w \in \tilde V: (v,w) \in \tilde E} \hspace{-3mm} f_{vw} = 0, \forall v \in \tilde V, \label{eq:node0}\\
  & f_{ww'} \leq \mu_w, ~~~ \forall w \in V, \label{eq:cap0}\\
  & f_{uv} + f_{u'v'} \leq \mu_{uv}, ~~~ \forall (u,v) \in E, \label{eq:cap1}\\
  & f_{uv} \geq 0, f_{u'v'} \geq 0, ~~~ \forall (u,v) \in E, \nonumber\\
  & f_{ww'} \geq 0, ~~~ \forall w \in V. \nonumber
\end{align}

Flow conservation constraints are guaranteed by \eqref{eq:node0}. Computation capacity constraints are guaranteed by \eqref{eq:cap0} for each computation node. Communication capacity constraints are guaranteed by \eqref{eq:cap1} for each communication link. The linear program has $O(|E|)$ variables and $O(|E|)$ constraints, which has a significantly smaller size compared with the path-based formulation. To conclude, the maximum flow can be computed by the linear program in polynomial time.

\subsubsection{Evaluation of minimum cuts}\label{sc6:cut}
Recall that an \st flow can be interdicted by removing either communication or computation resources, or a combination of both. We first develop an integer program to compute the minimum joint communication and computation cut, which can be easily modified to compute the minimum communication cut and the minimum computation cut. The formulation is based on disconnecting $s$ and $t'$ in the layered graph, which equivalently reduces the \st flow to zero in the original graph by Lemma~\ref{th:layercut}.

We use different \emph{node potentials} to indicate the separation of nodes in $\tilde{V}$ after removing the joint communication and computation cut. The potential of a node can be interpreted as its distance to $t'$, where the edges in the cut have unit lengths and the remaining edges have zero lengths. Let $p_v$ indicate the potential of a node. Suppose that $p(s) - p(t') \geq 1$. There is no path between $s$ and $t'$ that only consists of zero-length edges. Therefore, $s$ and $t'$ are disconnected after removing the edges in the cut. Let $y_{uv}$ indicate whether link $(u,v)$ is removed. Let $y_w$ indicate whether the computation resource at node $w$ is removed. The node potential never decreases along a connected path, guaranteed by constraints \eqref{eq:edge1}, \eqref{eq:edge2}, and \eqref{eq:node} when $y_{uv}=0$ and $y_w=0$. Disconnected nodes may have different potentials, guaranteed by the same constraints when $y_{uv}=1$ or $y_w=1$. If all the constraints are satisfied, $s$ and $t'$ are disconnected, since the potential cannot decrease along a connected path. The cut include the communication links where $y_{uv} = 1$ and computation nodes where $y_w = 1$. Notice that if link $(u,v)$ is removed, no flow can pass through either $(u,v)$ or $(u',v')$. Therefore, $y_{uv}$ appears in both Eqs. \eqref{eq:edge1} and \eqref{eq:edge2}.

\begin{align}
  \min ~~~& ~~~\sum_{(u,v) \in E} \mu_{uv} y_{uv} + \sum_{w \in V} \mu_{w} y_w \label{eq:cutpoly}\\
  \text{s.t.} ~~~
  & p_v - p_u + y_{uv} \geq 0, ~~~ \forall (u,v) \in E, \label{eq:edge1}\\
  & p_{v'} - p_{u'} + y_{uv} \geq 0, ~~~ \forall (u,v) \in E, \label{eq:edge2}\\
  & -p_w + p_{w'} + y_w \geq 0, ~~\forall w \in V, \label{eq:node}\\
  & p_s - p_{t'} \geq 1, \nonumber \\
  & y_{uv} \in \{0,1\}, ~~~ \forall (u,v) \in E, \nonumber\\
  & y_w \in \{0,1\}, ~~\forall w \in V. \nonumber
\end{align}

To obtain the minimum computation cut, it suffices to set $y_{uv} = 0$ for all $(u,v) \in E$, and compute the optimal solution to the integer program. To obtain the minimum communication cut, it suffices to set $y_w = 0$ for all $w \in V$, and then compute the optimal solution to the integer program.

Since it is inefficient to compute the optimal solution of an integer program, we next develop a polynomial-time algorithm for evaluating the minimum computation cut, and approximation algorithms for evaluating the minimum communication cut and the joint cut.

\emph{Minimum computation cut: }
Since a flow needs to be processed by computation nodes along the paths from the source to the destination, removing all the computation resources along \st paths is sufficient and necessary to reduce the flow to zero. Such computation resources can be identified by computing the intersection of the set of nodes reachable from the source and the set of nodes that can reach the destination. Both sets can be computed by depth first search. The algorithm is summarized as follows, with time complexity $O(|E|)$.

\begin{algorithm}[h]
\caption{Algorithm for evaluating the minimum computation cut for an \st pair}
\begin{enumerate}
\item Compute the set of nodes $V_s$ such that there exists at least one path from $s$ to every node in $V_s$.
\item Compute the set of nodes $V_t$ such that there exists at least one path from every node in $V_t$ to $t$.
\item The minimum computation cut for the \st pair is $V_s \cap V_t$.
\end{enumerate}
\label{alg:comp}
\end{algorithm}

\emph{Minimum communication cut: }
If there is a single computation node $u$, then the minimum communication cut is the minimum of 1) the minimum cut that disconnects $s$ and $u$, and 2) the minimum cut that disconnects $u$ and $t$.

However, if there is more than one computation node, computing the minimum communication cut is NP-hard (Lemma \ref{th:comm}). Besides the integer program \eqref{eq:cutpoly}, we develop a 2-approximation algorithm, which runs in polynomial time and outputs a communication cut whose value is at most twice the minimum communication cut.

\begin{algorithm}[h]
\caption{Approximation algorithm for the minimum communication cut for an \st pair}
\begin{enumerate}
\item Construct a layered graph. Assign an arbitrarily high cost to every link across two layers. Assign $\mu_{uv}$ cost to each of the links $(u,v)$ and $(u',v')$.
\item Compute a minimum cut $C$ that separates $s$ and $t'$.
\item The communication cut is given by links $\{(u,v)| (u,v) \in C \text{~or~} (u',v') \in C \}$.
\end{enumerate}
\label{alg:comm}
\end{algorithm}

\begin{thm}\label{th:comm2}
  The communication cut obtained by Algorithm \ref{alg:comm} has a value that is at most twice the value of the minimum communication cut.
\end{thm}
\begin{proof}
  Let $S^*$ be the minimum \st communication cut, which has value $w$. By Lemma \ref{th:layercut}, in the layered graph, removing edges $S' = \{(u,v), (u',v')| (u,v) \in S^* \}$ disconnects $s$ and $t'$. The cost of $S'$ in the layered graph is $2w$.

  The minimum communication cut $C$ obtained by Algorithm \ref{alg:comm} has value at most $2w$, since $C$ is the minimum $s-t'$ cut in the layered graph and is no larger than $S'$. Note that $C$ contains no link across the two layers because every crossing link has an arbitrarily high cost. Consequently, $L=\{(u,v)| (u,v) \in C \text{~or~} (u',v') \in C \}$ is a communication cut in the original graph. Furthermore, the cost of removing links $L$ is no more than the cost of removing $C$. Therefore, $L$ has at most twice the value of the minimum communication cut.
\end{proof}

\emph{Minimum joint communication and computation cut: }
Algorithm \ref{alg:comm} can be modified to compute a joint cut whose value is at most twice the minimum joint cut. In the first step of Algorithm \ref{alg:comm}, instead of assigning an arbitrarily high cost to links across two layers, $\mu_w$ cost is assigned to link $(w,w')$. Using a similar proof to the proof of Theorem \ref{th:comm2}, we obtain the performance of the modified algorithm.
\begin{thm}
  The joint communication and computation cut obtained by the modified algorithm has a value that is at most twice the value of the minimum joint cut.
\end{thm}

\subsection{Relationship between max-flow and min-cuts}\label{sc6:mfmc}
The classical max-flow min-cut theorem states that the maximum amount of flow from $s$ to $t$ equals the value of the minimum cut that separates $s$ and $t$. In a computing network, we study the connections between maximum flow and various types of minimum cuts. Since either communication or computation could be the bottleneck to support a flow, the gap between the maximum flow and the minimum communication cut or the minimum computation cut could be arbitrarily large. For example, Fig. \ref{fig6:comm} illustrates that the gap between the minimum communication cut and the maximum flow can grow arbitrarily large as the communication bandwidth increases while the computation capacity stays the same, where the numbers adjacent to links and nodes represent the communication capacity and computation capacity, respectively. Similarly, Fig. \ref{fig6:comp} illustrates that the gap between the minimum computation cut and the maximum flow can be arbitrarily large.

\begin{figure}[h]
\centering
\includegraphics[width=.6\linewidth]{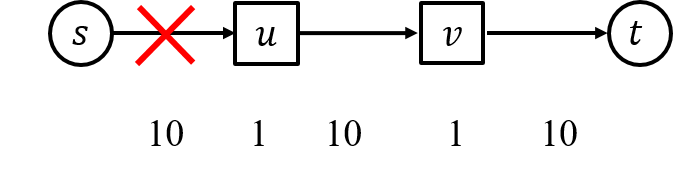}
\caption{Gap between the max flow and min communication cut: max flow $=$ 2, min communication cut $=$ 10.}
\label{fig6:comm}
\end{figure}

\begin{figure}[h]
\centering
\includegraphics[width=.59\linewidth]{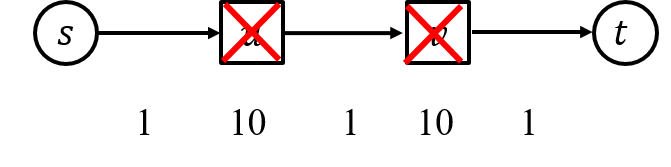}
\caption{Gap between the max flow and min computation cut: max flow $=$ 1, min computation cut $=$ 20.}
\label{fig6:comp}
\end{figure}

Since the joint communication and computation cut include pure communication cut and pure computation cut as special cases, the minimum joint cut is at most the smaller of the two pure cuts. In Fig. \ref{fig6:comm}, the minimum joint cut is 2, by removing the two units computation resources, while in Fig. \ref{fig6:comp}, the minimum joint cut is 1, by removing any one of the three communication links. Note that the joint cut can be smaller than both pure cuts. For example, consider two paths in parallel between $s$ and $t$, illustrated by Figs. \ref{fig6:comm} and \ref{fig6:comp}, respectively. The minimum joint cut is 3, while the minimum communication cut is 11 and the minimum computation cut is 22.

The following theorem bounds the gap between the maximum flow and the minimum joint cut.

\begin{thm}\label{th:gap}
  The minimum value (cf. \eqref{eq:cutpoly}) of the joint communication and computation cut is at most twice the maximum flow between a source-destination pair.
\end{thm}
\begin{proof}
  In the layered graph, the sum of flows on two copies of a communication link should not exceed the capacity of the link.
  By relaxing the capacity constraints, and restricting that the flow on each copy of the link should not exceed the capacity of the link, we obtain a \emph{modified} layered graph. Since the sum of flows in the two copies of a link is at most twice the link capacity, the capacity constraints in the original graph are satisfied by reducing the flow by half in the modified layered graph. Therefore, the maximum flow in the modified layered graph is at most twice the maximum flow in the original graph.

  The minimum cut in the modified layered graph is the same as the maximum flow in the modified layered graph. The minimum joint cut in the original graph is at most the minimum cut in the modified layered graph, since removing two copies of a link incurs double cost in the modified layered graph and a single cost in the original graph. Therefore, the minimum joint cut in the original graph is at most twice the maximum flow in the original graph.
\end{proof}

The gap is shown to be tight by the example in Fig. \ref{fig6:gap}. In this computing network, each link has capacity 2. Node $v$ is the only computation node with processing capacity 2. The maximum \st flow is 1, since the flow has to traverse link $(s,t)$ twice in order to be first processed and then delivered to $t$. Meanwhile, the minimum \st joint cut is 2.

\begin{figure}[h]
\centering
\includegraphics[width=.35\linewidth]{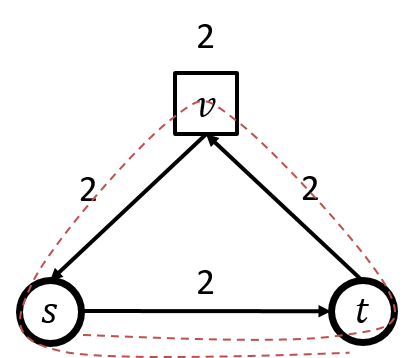}
\caption{Gap between the max flow and min joint cut: max flow $=$ 1, min joint cut $=$ 2.}
\label{fig6:gap}
\end{figure}

Unlike the classical communication network where links in a minimum \st cut are saturated by a maximum \st flow, a computing network may have links and nodes whose capacities are not saturated but still belong to the minimum cut. We next provide examples to support this observation.

\emph{Unsaturated node in minimum cut:} Consider a computing network represented by Fig. \ref{fig6:gap}, where the processing capacity at node $v$ is reduced to 1.5 and the other capacities remain the same. The maximum flow remains 1. The minimum joint cut is node $v$, which has value 1.5. However, only one unit processing capacity at $v$ is utilized by the maximum flow, and 0.5 unit processing capacity remains idle.

\emph{Unsaturated link in minimum cut:} In Fig. \ref{fig6:gaplink}, the maximum flow remains 1. The minimum cut is link $(u,t)$, which has capacity 1.5 and is not saturated by the maximum flow.

\begin{figure}[h]
\centering
\includegraphics[width=.6\linewidth]{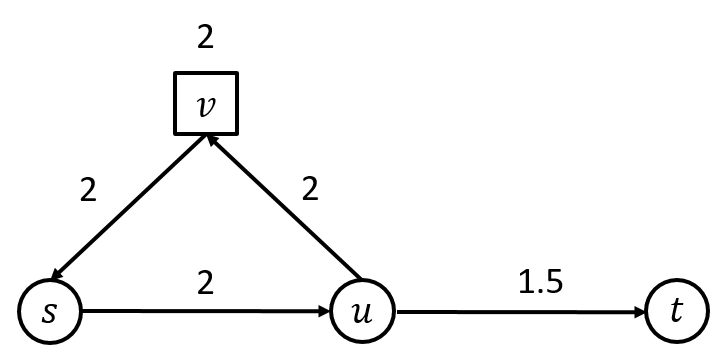}
\caption{Link in min cut may not be saturated by max flow.}
\label{fig6:gaplink}
\end{figure}

\section{Flow interdiction}
\label{sc6:inter}
In this section, we study network flow interdiction problems in a computing network. The objective is to minimize the maximum \st flow by removing communication links and computation resources using a given budget. We first discuss the complexity of flow interdiction problems in a computing network, and then develop mathematical programs and approximation algorithms to compute optimal and near-optimal interdiction strategies.

\subsection{Flow interdiction variants and complexity}
There are two major types of network flow interdiction problems -- binary interdiction and partial interdiction. For binary interdiction, the capacity of an interdicted link or a node is removed in full, at a fixed cost. For partial interdiction, the capacity can be removed by a fraction at a fractional cost.

We start by reviewing the complexity of network flow interdiction problems in a communication network, based on the classical flow network model, which imply the complexity of interdiction problems in a computing network. Suppose that the interdiction cost $c_{uv}$ for a link $(u,v)$ is arbitrary, and is independent of its capacity $\mu_{uv}$. The binary interdiction is NP-hard for a communication network, by a reduction from the knapsack problem \cite{wood1993deterministic,phillips1993network}.  Moreover, the binary interdiction remains NP-hard even if every link has one unit interdiction cost \cite{wood1993deterministic}. The optimal set of interdicted links belong to some \emph{minimal} cut, and the optimal partial interdiction strategy to attack a minimal cut is greedy in the decreasing value of $\mu_{uv}/c_{uv}$. Partial interdiction, on a network with unit link interdiction cost and an integer interdiction budget, reduces to binary interdiction, and is therefore NP-hard \cite{phillips1993network}.

One special case for network interdiction is that the link interdiction cost equals its capacity. The binary interdiction problem remains NP-hard, by a reduction from the subset sum problem. On the other hand, the partial interdiction problem can be solved in polynomial time, and the optimal interdiction is to interdict the links in any minimum cut. The maximum flow is $\max\{C - B,0\}$ after the partial interdiction, where $C$ is the minimum cut value and $B$ is the interdiction budget. The same solution can be extended to the problem where link interdiction cost is proportional to its capacity (i.e., $c_{uv} = \alpha \mu_{uv}$, where  $\alpha$ is identical for all links).

Since the flow interdiction problem in a computing network includes the flow interdiction problem in a communication network as a special case, the interdiction problem in a computing network is NP-hard for binary interdiction, and for partial interdiction with arbitrary costs. Nevertheless, the optimal partial interdiction is non-trivial even if the interdiction cost equals capacity. The optimal interdiction may not be the minimum cut. For example, in Fig. \ref{fig6:gaplink}, for budget $B < 1$, the optimal strategy is to interdict link $(s,u)$, while for budget $B > 1$, the optimal strategy is to interdict link $(u,t)$. The maximum flow under the optimal interdiction is:
\[
  f = \left \{
  \begin{tabular}{lll}
  $1 - 0.5B$,& & $B \leq 1$ \\
  $1.5 - B$,& & $1 < B \leq 1.5$ \\
  $0$,&& $B > 1.5$
  \end{tabular}
  \right.
\]

We further discuss the properties of the optimal partial interdiction when attack cost equals removed capacity. A flow can be decomposed into computation paths. After cycle canceling of the unprocessed flow, and the processed flow, respectively, a flow on a computation path traverses the same link no more than twice, once before processing and once after processing. By removing link capacity, the rate of flow decrease is 0, 0.5, or 1. By removing node capacity, the rate of flow decrease is 0 or 1. The maximum flow after the optimal partial interdiction is a piecewise linear function in the budget.

The function is neither convex nor concave in general, as shown by an example illustrated by Fig. \ref{fig:interdiction}. The maximum flow before interdiction is 2, which is carried by computation pathes $(\{s-w-t\},w)$ and $(\{s-u-v-s-u-t\},v)$. For interdiction budget $B \leq 1$, the optimal strategy is to attack link $(s,w)$. For $1 < B \leq 2$, the optimal strategy is to attack link $(s,u)$ in addition to $(s,w)$. For $2 < B \leq 2.5$, the optimal strategy is to attack link $(u,t)$ in addition to $(s,w)$. The rates of max flow decrease are 1, 0.5, and 1, respectively.
\begin{figure}[h]
\centering
  \begin{subfigure}[b]{0.3\textwidth}
    \centering
    \includegraphics[width=\textwidth]{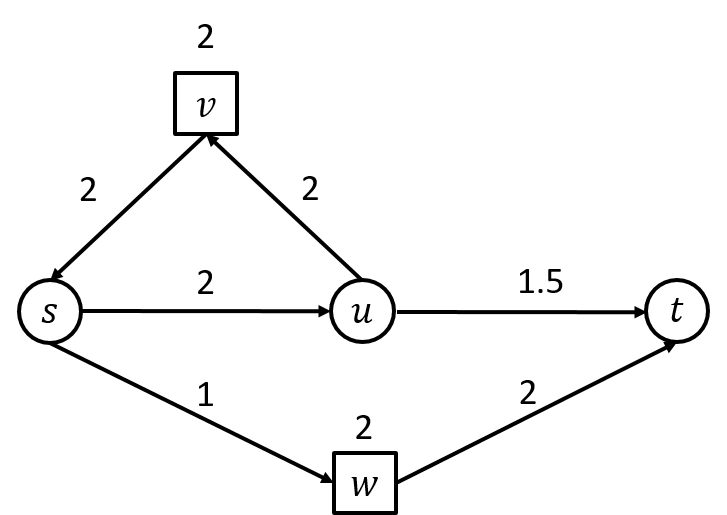}
  \end{subfigure}
  \begin{subfigure}[b]{0.35\textwidth}
    \centering
    \includegraphics[width=\textwidth]{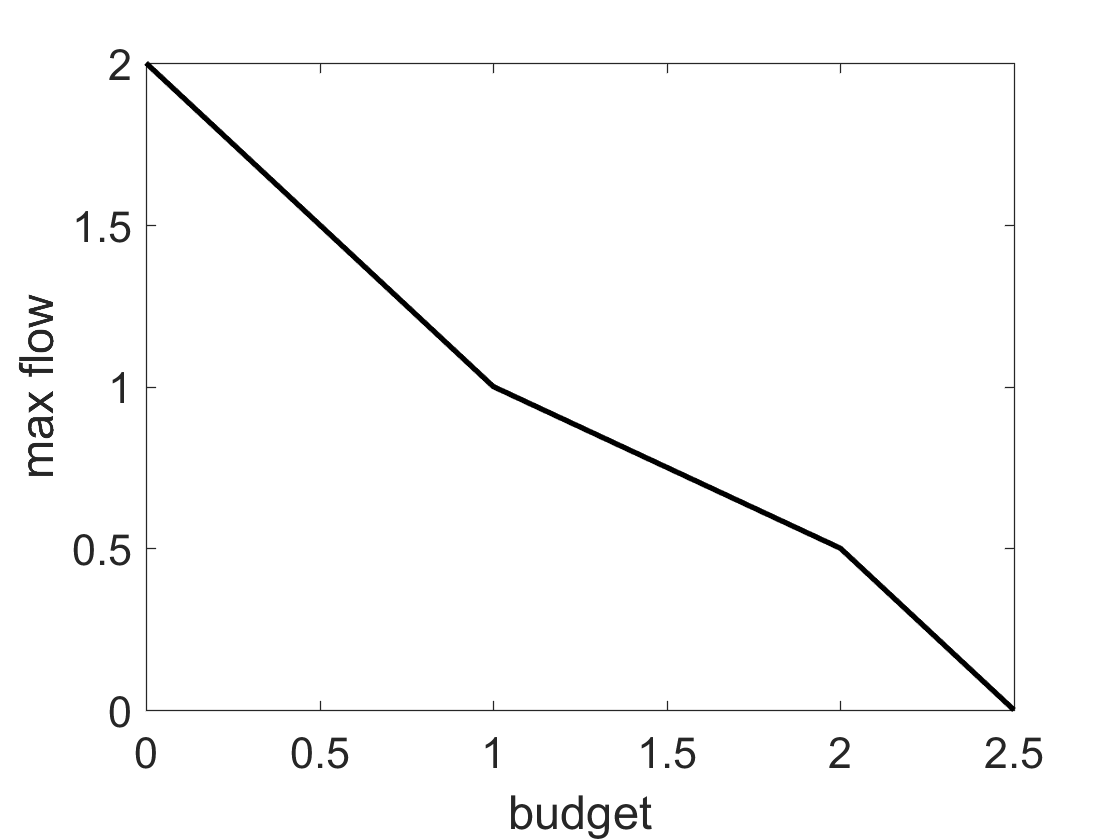}
  \end{subfigure}
  \caption{Example for max flow after optimal interdiction.}\label{fig:interdiction}
\end{figure}


\subsection{Exact solutions}
We develop mathematical programs to compute the optimal interdiction. The key is to transform the minimax problem (i.e., minimizing the maximum flow) to a minimization problem. Using linear programming duality, the maximum flow is equivalent to the minimum cut without integrality constraints. For binary interdiction, let $z_{uv}$ indicate whether link $(u,v)$ is removed, and let $z_w$ indicate whether the computation resource at node $w$ is removed. Let $c_{uv}$ denote the cost of removing link $(u,v)$. Let $c_w$ denote the cost of removing the computation resource at node $w$. Let $B$ denote the interdiction budget. The budget constraint is guaranteed by Eq. \eqref{eq:budget}. The objective \eqref{eq:interdiction} minimizes the maximum flow after interdiction. Informally, $\mu_{uv}\beta_{uv}$ indicates the amortized amount of flow contributed by link $(u,v)$, which is zero either by removing $(u,v)$ (i.e., $z_{uv} = 1$) or if link $(u,v)$ is not in a cut (i.e., $p_v = p_u$ and $p'_v = p'_u$). Similarly, $\mu_w \beta_w$ indicates the amortized flow contributed by computation node $w$. The detailed derivation of this formulation can be found in the Appendix.


\begin{align}
  \min ~~~& ~~~\sum_{(u,v) \in E} \mu_{uv} \beta_{uv} + \sum_{w \in V} \mu_w \beta_w \label{eq:interdiction}\\
  \text{s.t.} ~~~
  & p_v - p_u + \beta_{uv} + z_{uv} \geq 0, ~~~ \forall (u,v) \in E \nonumber \\
  & p_{v'} - p_{u'} + \beta_{uv} + z_{uv} \geq 0, ~~~ \forall (u,v) \in E \nonumber \\
  & -p_w + p_{w'} + \beta_w + z_w \geq 0, ~~\forall w \in V \nonumber\\
  & p_s - p_{t'} \geq 1, \nonumber \\
  & \sum_{(u,v) \in E} c_{uv}z_{uv} + \sum_{w \in V} c_w z_w \leq B, \label{eq:budget} \\
  & 0 \leq \beta_{uv} \leq 1, z_{uv} \in \{0, 1\}, ~~~ \forall (u,v) \in E, \nonumber\\
  & 0 \leq \beta_w \leq 1, z_w \in \{0, 1\}, ~~~\forall w \in V. \nonumber
\end{align}


We next develop a bilinear program to compute the optimal partial interdiction. We use the same variables to represent costs and capacities as in the integer linear program \eqref{eq:interdiction}, except that the variables $z_{uv}$ and $z_w$ now denote the fraction of removed link and node capacities. In Eq. \eqref{eq:flow2}, $\mu_{uv}(1 - z_{uv})$ denotes the remaining transmission capacity of link $(u,v)$, and $\mu_w (1 - z_w)$ denotes the remaining processing capacity at node $w$. The objective is the dual of the maximum flow Eq. \eqref{eq:flowpoly} with reduced capacities after interdiction. Although there is no integral constraints, the bilinear program is difficult to solve since the objective function is non-convex.
\begin{align}
  \min ~~~& ~~~\sum_{(u,v) \in E} \mu_{uv}(1 - z_{uv}) y_{uv} + \sum_{w \in V} \mu_w (1 - z_w) y_w \label{eq:flow2}\\
  \text{s.t.} ~~~
  & p_v - p_u + y_{uv} \geq 0, ~~~ \forall (u,v) \in E, \nonumber \\
  & p_{v'} - p_{u'} + y_{uv} \geq 0, ~~~ \forall (u,v) \in E, \\
  & -p_w + p_{w'} + y_w \geq 0, ~~\forall w \in V, \nonumber \\
  & p_s - p_{t'} \geq 1,  \nonumber \\
  & \sum_{(u,v) \in E} c_{uv}z_{uv} + \sum_{w \in V} c_w z_w \leq B, \nonumber \\
  & 0 \leq y_{uv} \leq 1, 0 \leq z_{uv} \leq 1, ~~~ \forall (u,v) \in E, \nonumber \\
  & 0 \leq y_w \leq 1, 0 \leq z_w \leq 1, ~~\forall w \in V.\nonumber
\end{align}

\subsection{Approximation algorithms}
We develop approximation algorithms based on the sensitivity analysis of the linear program \eqref{eq:flowpoly}. Namely, we study the change of the maximum flow under the changes of link capacity $\mu_{uv}$ and node capacity $\mu_w$. Our algorithms are in contrast with previous algorithms for classical flow interdiction based on minimizing the min-cut \cite{phillips1993network, wood1993deterministic, mccormick2014discrete}. Instead, our algorithms directly work with the max-flow, and the performance is not deteriorated by the non-zero gap between max-flow and min-cut in a computing network.

We start by considering the case where the attack cost equals the removed capacity. The \emph{shadow price} associated with a constraint in a linear program is the rate of change of the objective for one unit change of the right-hand side value of the constraint. Therefore, the shadow price associated with constraint \eqref{eq:cap0} represents the rate of max-flow decrease for each unit processing capacity decrease at node $w \in V$. The shadow price associated with constraint \eqref{eq:cap1} represents the rate of max-flow decrease for each unit transmission capacity decrease at link $(u,v) \in E$. Although the shadow price is a local property for a small change of the right-hand side of the constraint, it gives a conservative estimate of the impact on reducing the max-flow by reducing the capacity of a node or a link. The reason is that the rate of max-flow decrease is a monotone non-decreasing function as the capacity decreases, due to Lemma \ref{th:concave}.

\begin{lem}\label{th:concave}
  The maximum \st flow $F^*(\bm{\mu})$ given by linear program \eqref{eq:flowpoly} is a concave function of $\bm{\mu}$, where $\bm{\mu}$ is a vector representing link and node capacities.
\end{lem}
\begin{proof}
  Let $F^*(\bm{\mu}^i)$ be the maximum \st flow given a capacity vector $\bm{\mu}^i$, $\forall i \in \{1,2\}$. Let $\bm{f}^*(\bm{\mu}^i)$ denote the flow vector on each link that support the maximum flow, given capacity $\bm{\mu}^i$. Consider a capacity vector $\bm{\mu}^3 = \alpha \bm{\mu}^1 + (1 - \alpha) \bm{\mu}^2$, where $0 \leq \alpha \leq 1$. The flow vector $\bm{f}(\bm{\mu}^3) = \alpha \bm{f}^*(\bm{\mu}^1) + (1 - \alpha) \bm{f}^*(\bm{\mu}^2)$ is a feasible solution to the linear program \eqref{eq:flowpoly}, which supports $F(\bm{\mu}^3) = \alpha F^*(\bm{\mu}^1) + (1 - \alpha) F^*(\bm{\mu}^2)$ flow from $s$ to $t$. The maximum flow $F^*(\bm{\mu}^3)$ is at least $F(\bm{\mu}^3)$. Therefore, $F^*(\bm{\mu})$ is a concave function of $\bm{\mu}$.
\end{proof}

We propose Algorithm \ref{alg:greedy} that greedily computes the attack using the shadow price information in linear program \eqref{eq:flowpoly}.

\begin{algorithm}[h]
\caption{Greedy binary interdiction using budget $B$ on a network where interdiction cost equals removed capacity}
\begin{enumerate}
\item Solve linear program \eqref{eq:flowpoly} and obtain shadow prices for constraints \eqref{eq:cap0} and \eqref{eq:cap1}.
\item Choose a link or a node whose capacity is no more than $B$ and is associated with a constraint that has the largest shallow price. Denote the capacity by $\mu^*$.
\item Update $B$ by $B - \mu^*$. Repeat Step 1 until $B \leq 0$.
\end{enumerate}
\label{alg:greedy}
\end{algorithm}

Algorithm \ref{alg:greedy} can be naturally extended to the partial interdiction case. In Step 2, a link or a node associated with the constraint that has the largest shadow price is chosen. All its capacity is removed if the remaining budget is sufficient, while partial capacity is removed otherwise. The other steps of the algorithm remain the same.

We then develop Algorithm \ref{alg:greedycost} that computes an attack strategy when the interdiction cost is arbitrary, and not necessarily equal to the removed capacity. The algorithm can also be extended to solve the partial interdiction problem, in the same manner as Algorithm \ref{alg:greedy}.

\begin{algorithm}[h]
\caption{Greedy binary interdiction using budget $B$ on a network with arbitrary interdiction cost}
\begin{enumerate}
\item Solve linear program \eqref{eq:flowpoly} and obtain shadow prices $\bm{q}$ for constraints \eqref{eq:cap0} and \eqref{eq:cap1}.
\item Choose a link or a node whose capacity is no more than $B$ and is associated with a constraint that has the largest $q_i \mu_i / c_i$, where $\mu_i$ is the capacity and $c_i$ is the interdiction cost.
\item Update $B$ by $B - \mu_i$. Repeat Step 1 until $B \leq 0$.
\end{enumerate}
\label{alg:greedycost}
\end{algorithm}

The drawback of Algorithm \ref{alg:greedycost} is that the shadow price is only dependent on the capacity, but not cost. Under arbitrary interdiction cost, it is possible that a cut has a small capacity (i.e., bottleneck for traffic) but a high interdiction cost, in which case it is wise to attack a cut that has a lower interdiction cost but may allow more traffic to go through. The shadow prices associated with links in larger cuts are always zero, since reducing their capacities by a small amount would not reduce the maximum flow.

To overcome this difficulty, we develop a \emph{cost-aware greedy algorithm}. The algorithm is based on the linear program \eqref{eq:flowpoly}, where the capacities $\bm{\mu}$ in the right-hand side of the constraints are replaced by interdiction costs $\bm{c}$. The remaining steps are identical to Algorithm \ref{alg:greedycost}. The reason of using $\bm{c}$ as the new capacities is that the linear program serves as an approximation of the minimum-cost cut by relaxing the integral constraints.



\section{Numerical results}
\label{sc6:num}
In this section, we provide numerical examples to illustrate the applications of our proposed metrics and algorithms to study computing network robustness. First, we study the robustness based on the Abilene network topology in Fig. \ref{fig6:map}, which has 11 nodes and 14 links. Since we study directed graphs throughout the paper, we consider each edge in the figure as bidirectional links. In the last part of the section, we compare the accuracy and running time of the algorithms on CenturyLink (Level 3) network in U.S., which has 170 nodes and 230 links.

\begin{figure}[h]
\centering
\includegraphics[width=.6\linewidth]{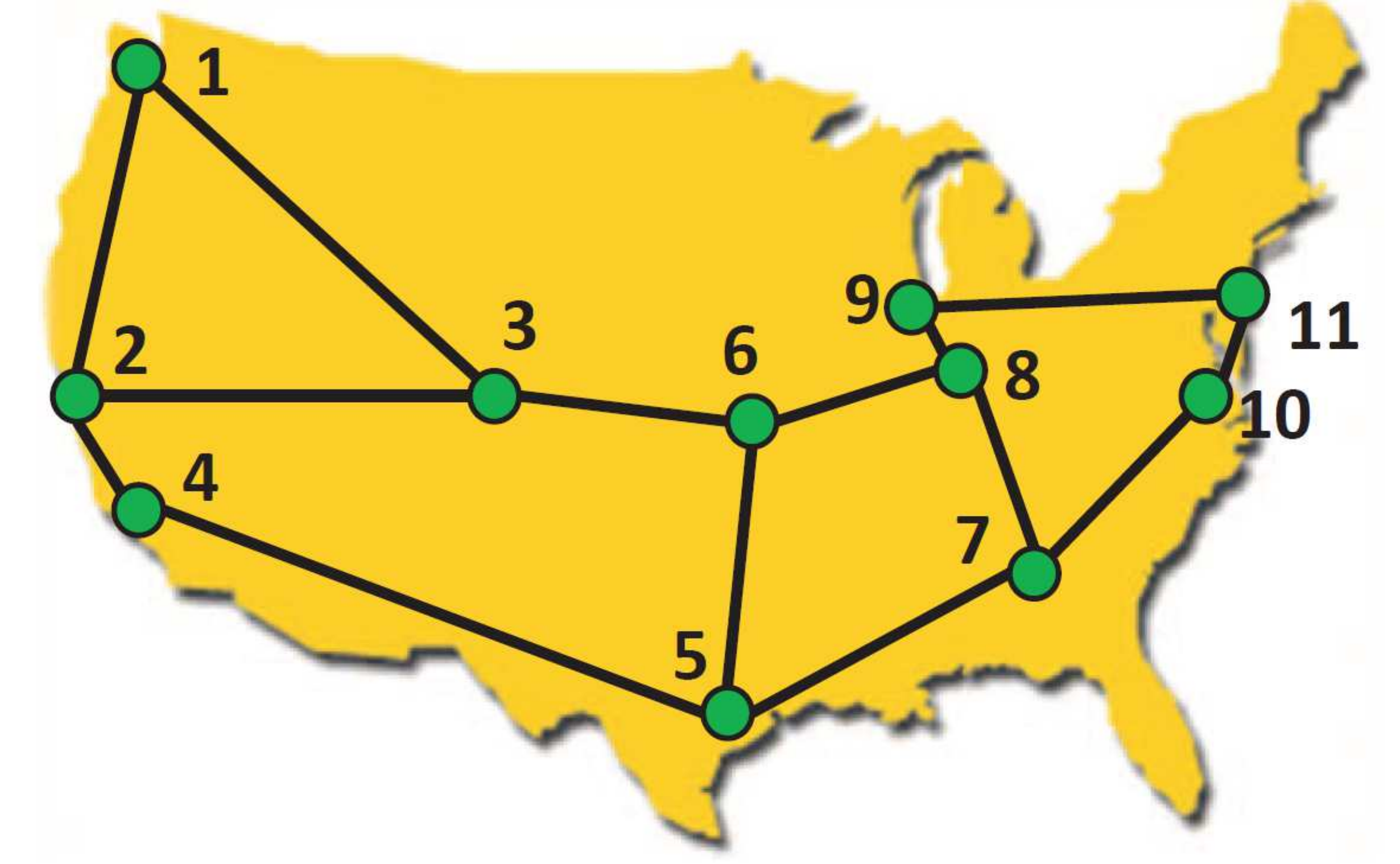}
\caption{Abilene network topology.}
\label{fig6:map}
\end{figure}

\subsection{Max-flow and min-cuts}
The maximum flow equals the minimum computation cut between an \st pair if computation resource is the bottleneck to support a network flow. Suppose that each directed link has transmission capacity 1, and that each of nodes 6 and 11 has processing capacity 0.5. The maximum flow between each pair of nodes is 1, which matches the value of minimum computation cut (i.e., nodes 6 and 11).

There may be a non-zero gap between the maximum flow and the minimum communication cut even if communication resource is the bottleneck to support a network flow. Suppose that each processing capacity of nodes 6 and 11 is increased to 5. The minimum communication cut for $s=8,t=7$ is 3, while the maximum flow is 2.5. The flow can be decomposed as follows. One unit flow is transmitted through $8-6-5-7$ and processed at 6. One unit flow is transmitted through $8-9-11-10-7$ and processed at 11. Half unit flow is transmitted through $8-7-5-6-8-7$ (or $8-7-10-11-9-8-7$) and processed at 6 (or 11). Part of the flow has to traverse link $8-7$ twice, once before processing and once after processing.

In the above two examples, the minimum joint cut equals the minimum of the pure communication cut and pure computation cut. By setting the processing capacity of nodes 6 and 11 to be 5 and 0.5, respectively, for $s=8,t=7$, the minimum joint cut is 2.5, smaller than both the minimum communication cut 3 and the minimum computation cut 5.5. In this example, the maximum \st flow is 2.25. One feasible decomposition of the flow is one unit flow through $8-6-5-7$ and processed at 6, half unit flow through $8-9-11-10-7$ and processed at 11, half unit flow through $8-9-11-10-7-5-6-8-7$ and processed at 6, and 0.25 unit flow through $8-7-5-6-8-7$ and processed at 6.

\subsection{Flow interdiction}
We then study flow interdiction using randomly generated capacities. For simplicity, the capacity of each link is independently and uniformly chosen from $(0,1)$. The capacity of each node is independently and uniformly chosen from $(0,0.1)$. 

First, we consider the network flow interdiction problem where the cost of interdiction equals the capacity. For $s=1, t=2$, the values of max-flow after optimal binary interdiction (solving the integer-linear program), approximate binary interdiction based on Algorithm \ref{alg:greedy}, and approximate partial interdiction based on an extension of Algorithm \ref{alg:greedy} are presented in Fig. \ref{fig6:int1}. The curve for the optimal binary interdiction is smooth, because computation resource is the bottleneck for the flow from node 1 to node 2 and computation capacity has finer granularity due to the small random number generation range. We observe that the approximate binary interdiction algorithm has good performance. Moreover, the approximate partial interdiction algorithm gives exact solutions, since the slope of the red curve is $-1$ and thus there is a unit max-flow decrease by removing each unit capacity, which is the maximum possible decrease.

For $s=1, t=10$, the values of max-flow after optimal binary interdiction, approximate binary interdiction, and approximate partial interdiction are presented in Fig. \ref{fig6:int2}. The curve for optimal binary interdiction has larger steps, because communication resource is the bottleneck for the flow from node 1 to node 10 and the cost of removing a link is relatively high. The steps in the curve illustrates that the interdiction problem has the same nature as the knapsack problem where the knapsack size represents the interdiction budget and item sizes represent interdiction cost (i.e., link capacity). We observe that Algorithm \ref{alg:greedy} and its extension still have good performance for both binary and partial interdictions. 

\begin{figure}[h]
\centering
    \includegraphics[width=.8\linewidth]{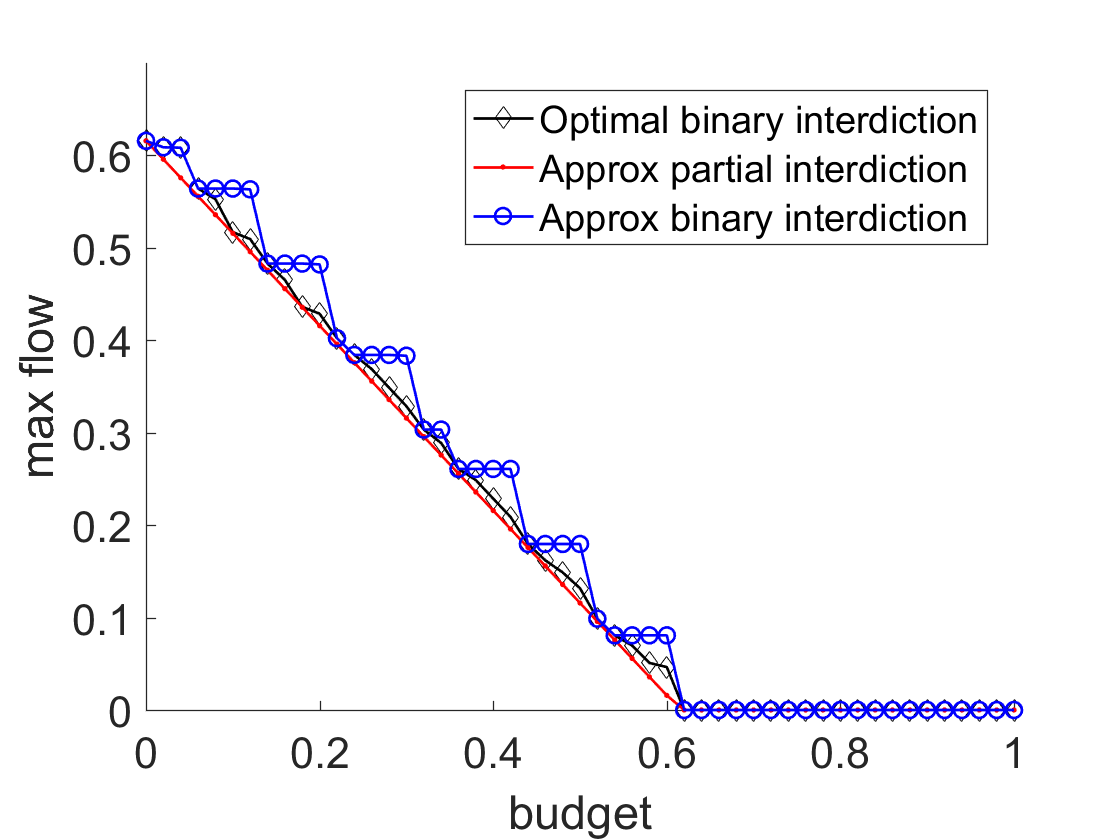}
    \caption{Cost equals capacity, $s=1,t=2$.}
    \label{fig6:int1}
\end{figure}
 \begin{figure}[h]
\centering
    \includegraphics[width=.8\linewidth]{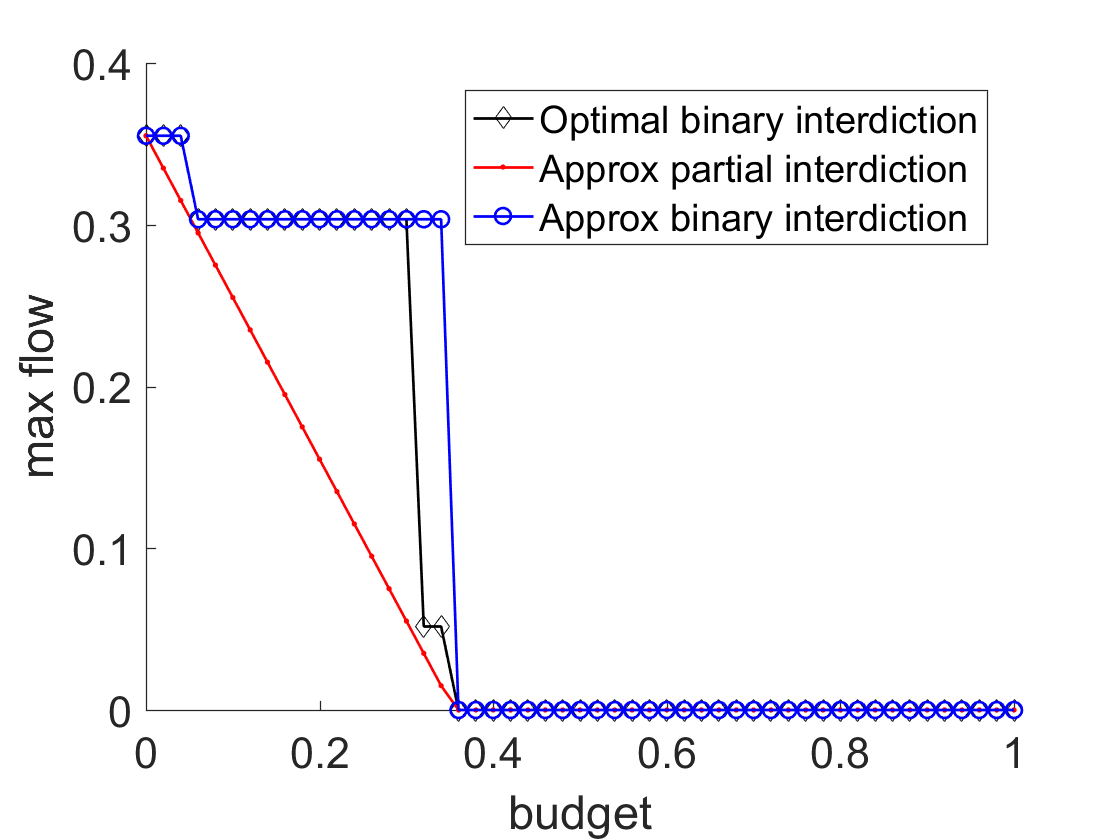}
    \caption{Cost equals capacity, $s=1,t=10$.}
    \label{fig6:int2}
  \end{figure}

Then, we study the performance of the interdiction algorithms under arbitrary interdiction costs. The cost of removing each link is independently and uniformly chosen from $(0,1)$. The cost of removing the computation resource at each node is independently and uniformly chosen from $(0,0.1)$.

For $s=1, t=2$, the values of max-flow after the optimal binary interdiction, approximate binary interdiction using Algorithm \ref{alg:greedycost}, and approximate partial interdiction using an extension from Algorithm \ref{alg:greedycost} are presented in Fig. \ref{fig6:int3}. The curve for optimal binary interdiction is steeper for small budgets compared with Fig. \ref{fig6:int1}, since it is possible to remove large computation resource at small cost due to the independence between cost and capacity. However, the performance of the greedy algorithms deteriorates. It is worth noting that the objective is not monotone under the greedy algorithm. The reason is that, in the greedy algorithm, a saturated link that has a high cost can be ruled out when the budget is small, which allows a larger cut that has a smaller cost to be removed. Similarly, the performance of the algorithms for $s=1,t=10$ is illustrated in Fig. \ref{fig6:int4}.

\begin{figure}[h]
\centering
    \includegraphics[width=.8\linewidth]{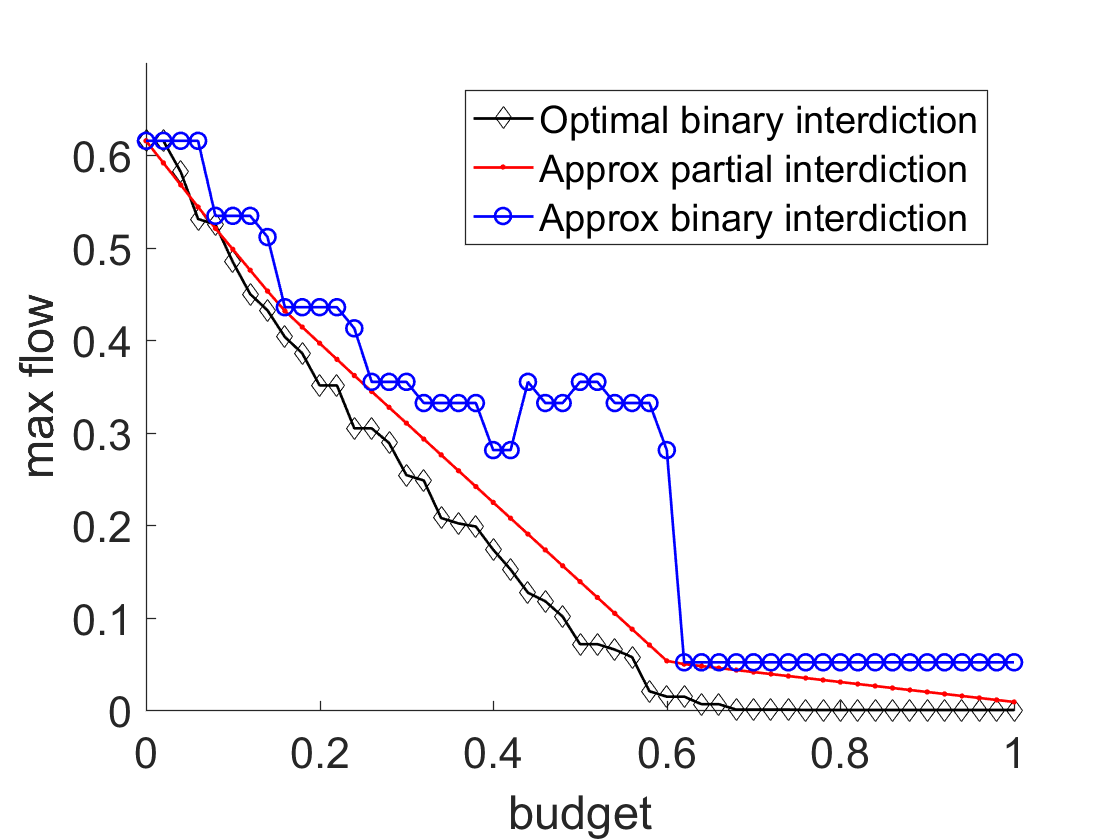}
    \caption{Cost independent of capacity, $s=1,t=2$.}
    \label{fig6:int3}
  \end{figure}
  \begin{figure}[h]
  \centering
    \includegraphics[width=.8\linewidth]{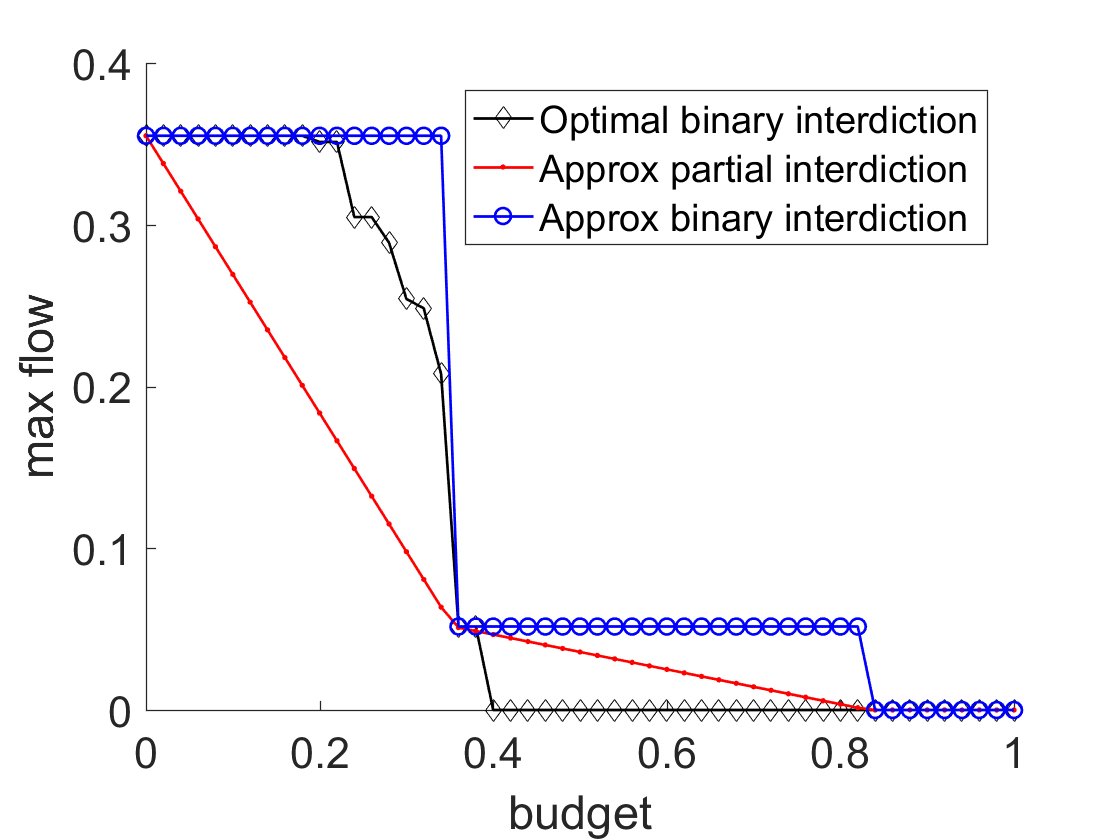}
    \caption{Cost independent of capacity, $s=1,t=10$.}
    \label{fig6:int4}
  \end{figure}

We next compare the greedy algorithm and the cost-aware greedy algorithm on network interdiction with arbitrary cost. We observe in Fig. \ref{fig6:compare} that the cost-aware greedy algorithm has significantly better performance for $s = 1, t = 2$. This can be explained by that a large number of computation nodes are attacked under the optimal strategy, and minimizing the attack cost becomes more important. However, the improvement is not significant for $s = 1, t = 10$, where a small number of links are attacked under the optimal strategy. The cost-aware greedy algorithm rely on the interdiction cost instead of capacity to compute the attacked links, and may attack links whose removal does not have a significant impact on the max-flow.

  \begin{figure}[h]
  \centering
    \includegraphics[width=.8\linewidth]{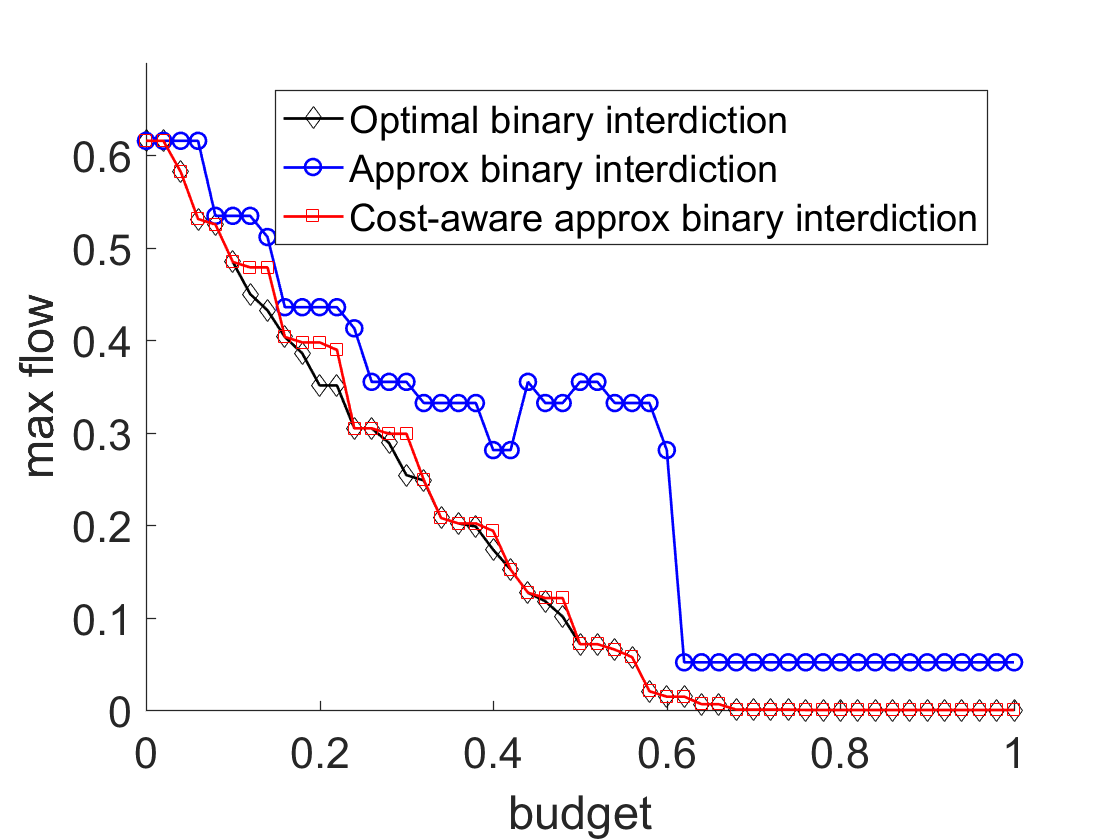}
    \caption{Comparison of greedy and cost-aware greedy algorithms, cost independent of capacity, $s=1,t=2$.}
    \label{fig6:compare}
  \end{figure}

\subsection{Scaling of the algorithms on larger network}
Finally, we study the performance of the algorithms by solving interdiction problems on the CenturyLink network illustrated by Fig. \ref{fig6:maplevel3}. We observe that the running time of solving the integer linear program \eqref{eq:interdiction} is sensitive to input parameters. The exact solution cannot be obtained within a pre-defined time limit for some problem instances. On the other hand, the greedy algorithms have good performance and have much shorter running time. The detailed results are reported below.

\begin{figure}[h]
\centering
\includegraphics[width=.8\linewidth]{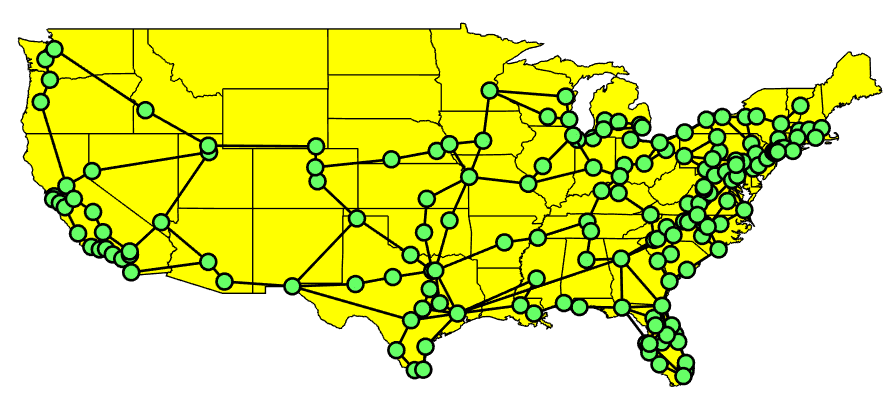}
\caption{CenturyLink (Level 3) network topology \cite{centuraylink}.}
\label{fig6:maplevel3}
\end{figure}

Each edge in Fig. \ref{fig6:maplevel3} represents bidirectional links. The capacity of a link is independently and uniformly chosen from $(0,10)$. The capacity of a node is independently and uniformly chosen from $(0,0.1)$. We first study the case where the interdiction cost equals removed capacity. By randomly choosing ten \st pairs, and using interdiction budget $B \in \{1,2,3,4,5,6\}$, the greedy interdiction computed by Algorithm \ref{alg:greedy} can be obtained in 5 seconds for every scenario. The integer linear program \eqref{eq:interdiction} fails to output an optimal solution within 10 minutes for two \st pairs. Among the scenarios where the optimal solutions are obtained, the running time ranges from 1 second to 4 minutes. The maximum flow after greedy interdiction is on average $7.7\%$ higher than the maximum flow after the optimal interdiction.

For arbitrary interdiction cost, we assume that the cost of removing each link is independently and uniformly chosen from $(0,10)$, and that the cost of removing the computation resource at each node is independently and uniformly chosen from $(0,0.1)$. Among ten randomly chosen \st pairs, the integer linear program \eqref{eq:interdiction} fails to output an optimal solution within 10 minutes for nine \st pairs. The running time of the greedy Algorithm \ref{alg:greedycost} and its cost-aware variant ranges from 1 to 40 seconds. The cost-aware greedy algorithm outperforms the greedy algorithm in 45 out of 60 scenarios, and yields $26\%$ lower maximum flow on average after interdiction.

\section{Conclusion}
\label{sc6:conclusion}
We studied the robustness of a distributed computing network where traffic flows require communication and computation resources to be transmitted and processed. We defined cut metrics to evaluate network robustness under the failures of communication and computation resources. We developed algorithms to evaluate the max-flow and the min-cuts, and showed a non-zero gap between them. Moreover, we developed algorithms for flow interdiction by removing communication and computation resources within a given budget.

\section*{Appendix}

\subsection{Complexity}
\begin{proof}[Proof of Lemma \ref{th:comm}]
  We first prove that obtaining the minimum \st communication cut is NP-hard if there are two computation nodes, by a reduction from exact cover by 3-sets. The reduction follows a similar proof in \cite{yannakakis1983cutting} that shows multicut is NP-hard.

  The exact cover by 3-sets problems is as follows. Given a set $X$ of $3q$ elements, and a collection $C$ of 3-element subsets of $X$, is there a subset $K \subseteq C$, such that every element in $X$ appears in exactly one member of $K$?

  We construct a graph from an instance of the exact cover by 3-sets problem. For each 3-set $c_i \in C$, there is a path $s_1 \rightarrow u_i \rightarrow v_i \rightarrow t_1$ from $s_1$ to $t_1$. The capacities of links $(s_1, u_i), (u_i,v_i), (v_i,t_1)$ are $k, 2, 1$, respectively. For each element $x \in X$, there is a path from $s_2$ to $t_2$. The path contains an edge $(u_i,v_i)$ if the 3-set $s_i$ contains $x$. All the other edges in the path from $s_2$ to $t_2$ have capacity $k$, except the edges $(u_i,v_i)$.

  Finally, the source node $s$ is connected to each of $s_1$ and $s_2$ through a link of capacity $k$. Each of the two nodes $t_1$ and $t_2$ is connected to the destination $t$ through a link of capacity $k$. The only two computation nodes are $s_2$ and $t_1$.

  Suppose the links adjacent to $s$ and $t$ are not removed. In order for a computation path to connect $s$ and $t$, either $s_1$ is connected to $t_1$, or $s_2$ is connected to both $t_1$ and $t_2$. If there exists an exact cover $K \subseteq C$ for $X$, a cut $S_c$ can be constructed as follows. The edge $(u_i,v_i)$ is in the cut if $s_i \in K$. The edge $(v_j,t_1)$ is in the cut if $s_j \notin K$. The value of the cut $S_c$ is $2q + (m - q) = m + q$, where $m = |C|$. This is the minimum cut that separates $t_1$ from $s_1$, and $\{t_1, t_2\}$ from $s_2$, for $k \geq 2m$. Therefore, $S_c$ is the minimum communication cut that disconnect all computation paths from $s$ to $t$.

  To conclude, the minimum communication \st cut is $m+q$ if and only if there exist exact cover by 3-sets for $X$. The reduction can be done in polynomial time, since there are $O(q+m)$ edges and vertices. The computation of the minimum communication \st cut is NP-hard.

  We illustrate the reduction using an example. Consider an exact cover by 3-sets problem where $X = \{1,2,3,4,5,6\}$, $C = \{c_1 = \{1,2,3\}, c_2 = \{1,2,4\}, c_3 = \{3,5,6\}\}$. In this example, $m = 3, q = 2$. There exist an exact cover $K = \{c_2, c_3 \}$ for $X$.
  The corresponding computing network is shown by Fig. \ref{fig6:np}. The path $s_1 \rightarrow u_i \rightarrow v_i \rightarrow t_1$ corresponds to the 3-set $c_i$, $\forall i \in \{1,2,3\}$. The path $s_2 \rightarrow u_1 \rightarrow v_1 \rightarrow u_2 \rightarrow v_2 \rightarrow t_2$ corresponds to elements 1 and 2 that appear in $c_1$ and $c_2$. The path $s_2 \rightarrow u_1 \rightarrow v_1 \rightarrow u_3 \rightarrow v_3 \rightarrow t_2$ corresponds to element 3 that appears in $c_1$ and $c_3$. The path $s_2 \rightarrow u_2 \rightarrow v_2 \rightarrow t_2$ corresponds to element 4 that appears in $c_2$. The path $s_2 \rightarrow u_3 \rightarrow v_3 \rightarrow t_2$ corresponds to elements 5 and 6 that appear in $c_3$. The thick edges each have capacity $k$. The numbers adjacent to the other edges indicate their capacities. The red edges $\{(u_2,v_2),(u_3,v_3),(v_1,t_1)\}$ illustrate the minimum computation cut. The value of the minimum computation cut is 5 = $m+q$.
\begin{figure}[h]
\centering
\includegraphics[width=.5\linewidth]{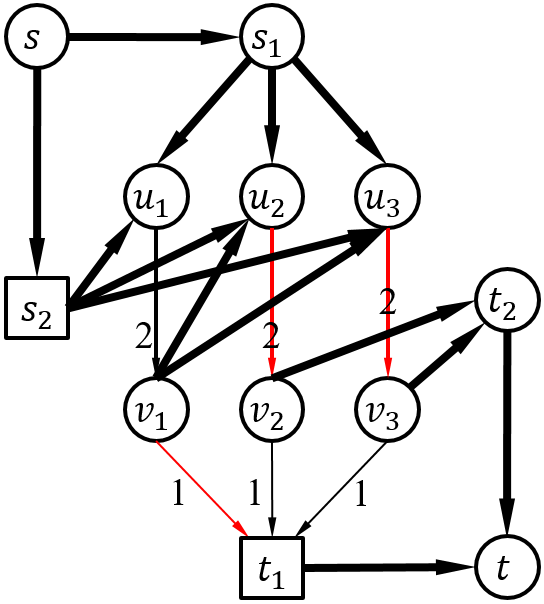}
\caption{Reduction from exact cover by 3-sets to minimum communication cut.}
\label{fig6:np}
\end{figure}
\end{proof}

\subsection{Correctness of the integer linear program for flow interdiction}
We prove the correctness of the integer linear program formulation for flow interdiction. The dual of the maximum flow (linear program \eqref{eq:flowpoly}) is equivalent to the minimum cut (integer program \eqref{eq:cutpoly}) without integrality constraints. Let $z_{uv}$ indicate whether link $(u,v)$ is removed. Let $z_w$ indicate whether the computation resource at node $w$ is removed. The maximum flow after removing links where $z_{uv}=1$ and computation resources at nodes where $z_w = 1$ is represented by Eq. \eqref{eq:flow2}. The mathematical program \eqref{eq:flow2} computes the maximum flow after the optimal interdiction with budget $B$.

\begin{align}
  \min ~~~& ~~~\sum_{(u,v) \in E} \mu_{uv}(1 - z_{uv}) y_{uv} + \sum_{w \in V} \mu_w (1 - z_w) y_w \label{eq:flow2}\\
  \text{s.t.} ~~~
  & p_v - p_u + y_{uv} \geq 0, ~~~ \forall (u,v) \in E, \nonumber \\
  & p_{v'} - p_{u'} + y_{uv} \geq 0, ~~~ \forall (u,v) \in E, \\
  & -p_w + p_{w'} + y_w \geq 0, ~~\forall w \in V, \nonumber \\
  & p_s - p_{t'} \geq 1,  \nonumber \\
  & \sum_{(u,v) \in E} c_{uv}z_{uv} + \sum_{w \in V} c_w z_w \leq B, \nonumber \\
  & 0 \leq y_{uv} \leq 1, z_{uv} \in \{0, 1\}, ~~~ \forall (u,v) \in E, \nonumber \\
  & 0 \leq y_w \leq 1, z_w \in \{0,1\}, ~~\forall w \in V.\nonumber
\end{align}

Since $z_{uv}$ and $z_v$ are binary, the objective can be equivalently represented by Eq. \eqref{eq:flow3}, by adding constraints Eqs. \eqref{eq:beta1}, \eqref{eq:beta2}, \eqref{eq:beta3}, \eqref{eq:beta4}. To see this, note that if $z_{uv} = 0$, $\mu_{uv} \beta_{uv} \geq \mu_{uv} y_{uv}$. In the optimal solution to the integer linear program \eqref{eq:flow3}, $\mu_{uv} \beta^*_{uv} = \mu_{uv} y^*_{uv}$, since $\mu_{uv} \geq 0$. If $z_{uv} = 1$, $y_{uv} - z_{uv} \leq 0$, and $\mu_{uv} \beta^*_{uv} = 0$ in the optimal solution. In both cases, $\beta^*_{uv} \leq 1$. Therefore, the objective $(1 - z_{uv}) y_{uv}$ can be transformed to $\mu_{uv} \beta_{uv}$. Similarly, the objective $(1 - z_w) y_w$ can be transformed to $\mu_w \beta_w$.  The objective Eq. \eqref{eq:flow3} exactly matches the objective Eq. \eqref{eq:flow2}.


\begin{align}
  \min ~~~& ~~~\sum_{(u,v) \in E} \mu_{uv} \beta_{uv} + \sum_{w \in V} \mu_w \beta_w \label{eq:flow3}\\
  \text{s.t.} ~~~
  & p_v - p_u + y_{uv} \geq 0, ~~~ \forall (u,v) \in E, \nonumber \\
  & p_{v'} - p_{u'} + y_{uv} \geq 0, ~~~ \forall (u,v) \in E, \nonumber \\
  & -p_w + p_{w'} + y_w \geq 0, ~~\forall w \in V, \nonumber \\
  & p_s - p_{t'} \geq 1,  \nonumber \\
  & \sum_{(u,v) \in E} c_{uv}z_{uv} + \sum_{w \in V} c_w z_w \leq B, \nonumber \\
  & \beta_{uv} \geq y_{uv} - z_{uv}, ~~~ \forall (u,v) \in E, \label{eq:beta1} \\
  & \beta_{w} \geq y_w - z_w, ~~\forall v \in V, \label{eq:beta2} \\
  & 0 \leq \beta_{uv} \leq 1, ~~~ \forall (u,v) \in E, \label{eq:beta3}\\
  & 0 \leq \beta_w \leq 1, ~~~\forall w \in V, \label{eq:beta4} \\
  & 0 \leq y_{uv} \leq 1, z_{uv} \in \{0, 1\}, ~~~ \forall (u,v) \in E, \nonumber \\
  & 0 \leq y_w \leq 1, z_w \in \{0,1\}, ~~\forall w \in V.\nonumber
\end{align}

Finally, we show that the integer linear program \eqref{eq:flow3} has the same optimal solution, if the constraints \eqref{eq:beta1} and \eqref{eq:beta2} are replaced by equality constraints. Suppose that in an optimal solution, $y^*_{uv} - z^*_{uv} \geq 0$. Then $\beta^*_{uv} = y^*_{uv} - z^*_{uv}$ holds in the optimal solution. If $y^*_{uv} - z^*_{uv} < 0$, $y^*_{uv}$ can be increased to $z^*_{uv}$ without violating any constraint and achieves the same cost, such that $\beta^*_{uv} = 0$. Therefore, the constraint \eqref{eq:beta1} can be replaced by an equality constraint. The same analysis holds for replacing constraint \eqref{eq:beta2} by an equality constraint. By replacing $y_{uv} = \beta_{uv} + z_{uv}$ in all the constraints, we obtain the integer linear programming formulation \eqref{eq:interdiction}.

\begin{rmk}
The network flow interdiction problem in the classical communication network was formulated as an integer linear program in \cite{wood1993deterministic}. We follow a similar approach that use linear programming duality to transform a minimax problem to a minimization problem. The key difference is that the classical minimum cut polytope is integral, and thus it is possible to restrict values of $p_v, y_{uv}$ to be binary in \cite{wood1993deterministic}. However, the polytope of Integer program \eqref{eq:cutpoly} is not integral. Thus, $p_v, y_{uv}$ may take fractional values, which complicates our analysis and makes it non-trivial to extend this formulation to study partial interdiction problems.
\end{rmk}

\bibliographystyle{IEEEtran}
\bibliography{computation}
\end{document}